\documentclass[12pt]{article}
\usepackage{stmaryrd}
\usepackage{amsmath}
\usepackage{amssymb}
\usepackage{comment}
\usepackage{mathrsfs}
\usepackage{graphicx}
\usepackage{graphics,natbib,anysize}
\usepackage[colorlinks, linkcolor=blue, citecolor=blue, urlcolor=blue, pdfborder={0 0 1}, citebordercolor={0 1 0}]{hyperref}
\usepackage[body={6.6in,9.5in}]{geometry}
\def\boxit#1{\vbox{\hrule\hbox{\vrule\kern6pt\vbox{\kern6pt#1\kern6pt}\kern6pt\vrule}\hrule}}

\newenvironment{proof}[1][Proof]{\begin{trivlist}
\item[\hskip \labelsep {\bfseries #1}]}{\end{trivlist}}
\newtheorem{lemma}{{Lemma}}
\newtheorem{theorem}{{Theorem}}
\newtheorem{assumption}{{Assumption}}

\newtheorem{remark}{Remark}

\newtheorem{example}{{Example}}
\newtheorem{corollary}{{Corollary}}

\newcommand{\pkg}[1]{{\fontseries{b}\selectfont #1}}

\newcommand{\mbR}{\mathbb{R}}

\newcommand{\1}{{\rm 1}\kern-0.24em{\rm I}}

\newcommand{\mQ}{\mathcal{Q}}

\newcommand{\mJ}{\mathcal{J}}

\newcommand{\bA}{\boldsymbol{A}}

\newcommand{\btheta}{\boldsymbol{\theta}}

\newcommand{\bE}{\boldsymbol{E}}

\newcommand{\bZ}{\boldsymbol{Z}}
\newcommand{\bM}{\boldsymbol{M}}
\newcommand{\bX}{\boldsymbol{X}}
\newcommand{\bI}{\boldsymbol{I}}

\newcommand{\bSigma}{\boldsymbol{\Sigma}}

\newcommand{\bx}{\boldsymbol{x}}

\newcommand{\bw}{\boldsymbol{w}}
\newcommand{\bbf}{\boldsymbol{f}}
\newcommand{\bbg}{\boldsymbol{g}}
\newcommand{\bg}{\boldsymbol{g}}

\newcommand{\bD}{\boldsymbol{D}}
\newcommand{\bmu}{\boldsymbol{\mu}}

\newcommand{\bone}{\boldsymbol{1}}
\newcommand{\bzero}{\boldsymbol{0}}

\newcommand{\bbeta}{\boldsymbol{\beta}}

\newcommand{\bgamma}{\boldsymbol{\gamma}}

\usepackage{color}

\title{Feature Augmentation via Nonparametrics and Selection (FANS) in High Dimensional Classification\thanks{Jianqing Fan is Frederick L. Moore Professor of Finance,
Department of Operations Research and Financial Engineering,
Princeton University, Princeton, NJ, 08544 (Email:
jqfan@princeton.edu). Yang Feng is Assistant Professor, Department of Statistics, Columbia University, New York, NY, 10027 (Email: yangfeng@stat.columbia.edu). Jiancheng Jiang is Associate Professor, Department of Mathematics and Statistics, University of North Carolina at Charlotte, Charlotte, NC, 28223 (Email: jjiang1@uncc.edu). Xin Tong is  Assistant Professor, Department of Data Sciences and Operations, University of Southern California, Los Angeles, CA, 90089 (Email: xint@marshall.usc.edu).
The financial support from
National Institutes of Health grants R01-GM072611 and R01GM100474-01 and National Science Foundation grants DMS-1206464 and DMS-1308566 is
greatly acknowledged. The authors thank the editor,
the associate editor, and referees for their constructive comments.}}
\author{Jianqing Fan, Yang Feng, Jiancheng Jiang and Xin Tong}
\date{}
\vspace{1.5cm}
\begin{document}
\maketitle

\begin{abstract}

We propose a high dimensional classification method that involves  nonparametric feature augmentation. Knowing that marginal density ratios are the most powerful univariate classifiers,  we use the ratio estimates to transform the original feature measurements. Subsequently, penalized logistic regression is invoked, taking  as input the newly transformed or augmented features. This procedure trains models equipped with local complexity and global simplicity, thereby avoiding the curse of dimensionality while creating a flexible nonlinear decision boundary.   The resulting method is called Feature Augmentation via Nonparametrics and Selection (FANS). We motivate FANS by generalizing the Naive Bayes model, writing the log ratio of joint densities as a linear combination of those of marginal densities.  It is related to generalized additive models, but has better interpretability and computability.  Risk bounds are developed for FANS.   In numerical analysis, FANS is compared with competing methods, so as to provide a guideline on its best application domain. Real data analysis demonstrates that FANS performs very competitively on benchmark email spam and gene expression data sets. Moreover, FANS is implemented by an extremely fast algorithm through parallel computing.

\end{abstract}\textbf{Keywords:}  density estimation, classification, high dimensional space, nonlinear decision boundary, feature augmentation, feature selection, parallel computing.

\pagenumbering{arabic} \setcounter{page}{1}

\newtheorem{prop}{\textbf{\sc{Proposition}}}
\newtheorem{ass}{\textbf{\sc{Assumption}}}
\newtheorem{theo}{\textbf{\sc{Theorem}}}
\newtheorem{lem}{\textbf{\sc{Lemma}}}
\newtheorem{defi}{\textbf{\sc{Definition}}}
\newtheorem{coro}{\textbf{\sc{Corollary}}}
\section{Introduction}\label{sec::introduction}
Classification aims to identify to which category a new observation belongs based on feature measurements. Numerous applications include spam detection, image recognition, and disease classification (using high-throughput data such as microarray gene expression and SNPs).  Well known classification methods include Fisher's linear discriminant analysis (LDA), logistic regression, Naive Bayes, $k$-nearest neighbors, neural networks, and many others. All these methods can perform well in the classical low dimensional settings, in which the number of features is much smaller than the sample size. However, in many contemporary applications, the number of features $p$ is large compared to the sample size $n$. For instance, the dimensionality $p$ in microarray data is frequently in thousands or beyond, while the sample size $n$ is typically in the order of tens. Besides computational issues, the central conflict in high dimensional setup is that the model complexity is not supported by limited access to data. In other words, the ``variance" of conventional models is high in such new settings, and even simple models such as LDA need to be regularized.  We refer to   \cite{Hastie.Tibshirani.ea.2009}  and \cite{Buhlmann.Geer.2011} for overviews of statistical challenges associated with high dimensionality.

In this paper, we propose a classification procedure \underline{FANS} (Feature Augmentation via Nonparametrics and Selection).  Before introducing the algorithm, we first detail its motivation. Suppose feature measurements and responses are coded by a pair of random variables $(\bX, Y)$, where $\bX\in\mathcal{X}\subset \mathbb{R}^p$ denotes the features and $Y\in \{0,1\}$ is the binary response. Recall that a classifier $h$ is a data-dependent mapping from the feature space to the labels. Classifiers are usually constructed to minimize the risk
$P(h(\bX)\neq Y)$.

Denote by $g$ and $f$ the class conditional densities respectively for class 0 and class 1, i.e., $(\bX|Y=0)\sim g$ and $(\bX|Y=1)\sim f$.
It can be shown that the Bayes rule is $\1(r(\bx)\geq 1/2)$, where $r(\bx)=E(Y|\bX=\bx)$. Let $\pi=P(Y=1)$, then
$$
r(\bx) = \frac{\pi f(\bx)}{\pi f(\bx)+(1-\pi)g(\bx)}\,.
$$
Assume for simplicity that $\pi=1/2$, then the oracle decision boundary is
$$\{\bx: f(\bx)/g(\bx)=1\}=\{\bx: \log f(\bx) - \log g(\bx)=0\}\,,$$
Denote by $g_1, \cdots, g_p$ the marginals of $g$, and by $f_1, \cdots, f_p$ those of $f$. Naive Bayes models assume that the conditional distributions of each feature given the class labels are independent, i.e.,
\begin{align}\label{eq::Naive Bayes}
  \log\frac{f(\bx)}{g(\bx)}=\sum_{j=1}^p\log\frac{f_j(x_j)}{g_j(x_j)}\,.
\end{align}
Naive Bayes is a simple approach, but it is useful in many high-dimensional settings. Taking a two class Gaussian model with a common covariance matrix,  \citet{Bickel.Levina.2004} showed that naively carrying out the Fisher's discriminant rule performs poorly due to diverging spectra. In addition, the authors argued that independence rule which ignores the covariance structure performs better than the Fisher's rule in some high-dimensional settings.
However,  correlation among features is usually an essential characteristic of data, and it can help classification under suitable models and with relative abundance of the sample. Examples in bioinformatics study can be found in \cite{Ackermann.Strimmer.2009}. Recently, \cite{Fan.Feng.ea.2011} showed that the independence assumption can lead to huge loss in classification power when correlation prevails, and proposed a Regularized Optimal Affine Discriminant (ROAD).  ROAD is a linear plug-in rule targeting directly on the classification error, and it takes advantages of the un-regularized pooled sample covariance matrix.

Relaxing the two-class Gaussian assumption in parametric Naive Bayes gives us a general Naive Bayes formulation such as \eqref{eq::Naive Bayes}.  However, this model also fails to capture the correlation, or dependence among features in general.  On the other hand, the marginal density ratios are the most powerful univariate classifiers and using them as features in  multivariate classifiers can yield very powerful procedures.
This consideration motivates us to ask the following question: are there advantages of combining these transformed features rather than untransformed feature? More precisely, we would like to learn a decision boundary from the following set
\begin{equation} \label{eq1}
\mathcal{D}=\left\{\bx:  \beta_0+\beta_1 \log \frac{f_1(x_1)}{ g_1(x_1)}+\cdots + \beta_p \log \frac{f_p(x_p)}{ g_p(x_p)} = 0, \beta_0, \cdots, \beta_p\in \mathbb{R}\right\}\,.
\end{equation}
(All coefficients are $1$ in the Naive Bayes model, so optimization is not necessary.)
For univariate problems, properly thresholding the marginal density ratio delivers the best classifier. 
In this sense, the marginal density ratios can be regarded as the best transforms of future measurements, and (\ref{eq1}) is an effort towards combining those most powerful univariate transforms to build more powerful classifiers.

This is in a similar spirit to the sure independence screening (SIS) in \cite{Fan.Lv.2008} where the best marginal predictors are used as probes for their utilities in the joint model. By combining these marginal density ratios and
optimizing over their coefficients $\beta_j$'s, we wish to build a good classifier that takes into account feature dependence.  Note that although our target boundary $\mathcal{D}$ is not linear in the original features, it is linear in the parameters $\beta_j$'s.  Therefore, any linear classifiers can be applied to the transformed variables.   For example, we can use logistic regression, one of the most popular linear classification rules. Other choices, such as SVM (linear kernel), are good alternatives, but we choose logistic regression for the rest of discussion.

\noindent Recall that logistic regression models the log odds by
\begin{align*}
  \log\frac{P(Y=1|\bX=\bx)}{P(Y=0|\bX=\bx)}=\beta_0+\sum_{j=1}^p\beta_jx_j\,,
\end{align*}
where the $\beta_j$'s are estimated by the maximum likelihood approach. We should note that  without explicitly modeling correlations,
logistic regression takes into account features' joint effects and levels a good linear combination of features as the decision boundary. Its performance is similar to LDA, but both models can only capture decision boundaries linear in original features. 

On the other hand, logistic regression might serve as a building block for the more flexible FANS algorithm. Concretely, if we know the marginal densities $f_j$ and $g_j$, and run logistic regression on the transformed features $\{\log (f_j(x_j)/g_j(x_j))\}$, we create a decision boundary nonlinear in the original features.  The use of these transformed features is easily interpretable:  one naturally combines the ``most powerful" univariate transforms (building blocks of univariate Bayes rules) $\{\log (f_j(x_j)/g_j(x_j))\}$ rather than the original measurements.
In special cases such as the two-class Gaussian model with a common covariance matrix, the transformed features are not different from the original ones. Some caution should be taken: if $f_j=g_j$ for some $j$, i.e., the marginal densities for feature $j$ are exactly the same, this feature will not have any contribution in classification.  Deletion like this might lose power, because features having no marginal contribution on their own might boost classification performance when they are used jointly with other features. In view of this defect, a variant of FANS augments the transformed  features with  the original ones.

Since marginal densities $f_j$ and $g_j$ are unknown, we need to first estimate them, and then run a penalized logistic regression (PLR) on the estimated transforms. Note that some regularization (e.g., penalization) is necessary to reduce model complexity in the high dimensional paradigm.  This two-step classification rule of feature augmentation via nonparametrics and  selection will be called \underline{FANS} for short. Precise algorithmic implementation of FANS is described in the next section.     Numerical results show that our new method excels in many scenarios, in particular when no linear decision boundary can separate the data well.

%

\begin{figure}[h!]\label{fig::toyex}
  \begin{center}
\caption{The median test  errors for Gaussian vs. mixture of Gaussian when the  training data size varies. Standard errors shown in the error bars.}
\includegraphics[scale=0.7]{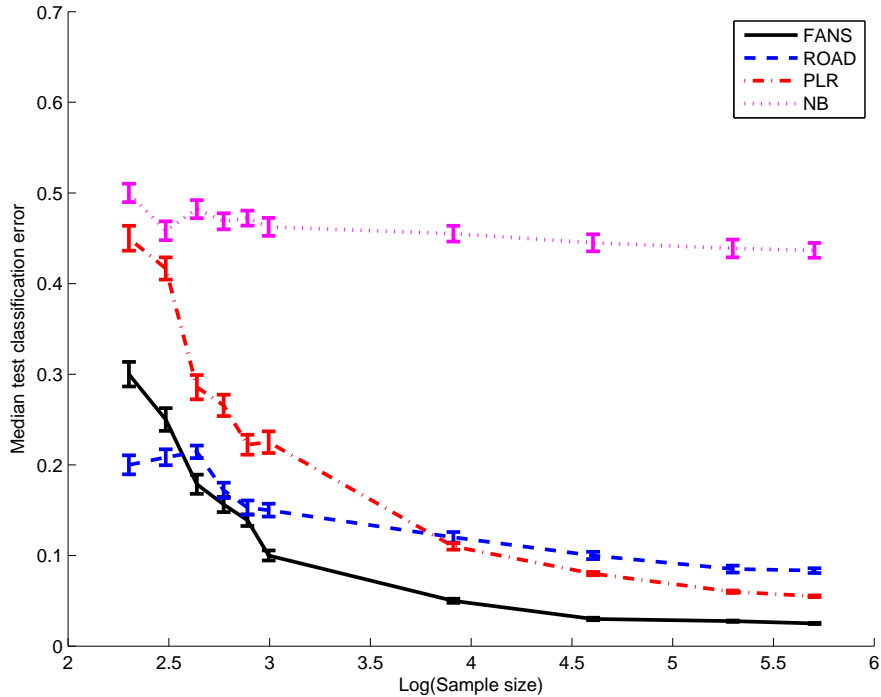}

\end{center}
\end{figure}

To understand where FANS stands compared to Naive Bayes (NB), penalized logistic regression (PLR), and  the regularized optimal affine discriminant (ROAD), we showcase a simple simulation example. In this example,  the choice is  between a multivariate Gaussian distribution and  some component－wise mixture of two multivariate Gaussian distributions:

Class 0: $ N((5\times\bone_{10}^T,\bzero_{p-10}^T)^T,\bSigma),$

Class 1: $\bw \circ N(\bzero_p,\bI_p)+(1-\bw)\circ N((6\times\bone_{10}^T,\bzero_{p-10}^T)^T,\bSigma), $
where $p=1000$, $\circ$ is the element-wise product between matrices, $\Sigma_{ii}=1$ for all $i=1, \cdots, p$, $\Sigma_{ij}=0.5$ for all $i,j=1, \cdots, p$  and $i\neq j$, and $\bw=(w_1,\cdots,w_p)^T$, in which $w_j\sim^{iid} \mbox{Bernoulli(0.5)}$.

The median classification error with the standard error shown in the error bar for 100 repetitions as a function of training sample size $n$ is rendered in  Figure \ref{fig::toyex}. This figure suggests  that increasing the sample size does not help NB  boost performance (in terms of the median classification error), because the NB model is severely biased in view of  significant correlation presence.  It is interesting to compare PLR with ROAD. ROAD is a more efficient approach when the sample size is small; however, PLR eventually performs better when the sample size becomes large enough. This is not surprising because the underlying true model is not two class Gaussian with a common covariance matrix. So the less ``biased" PLR beats ROAD on larger samples.  Nevertheless, even if ROAD uses a misspecified model, it still benefits from a specific model assumption on small samples.    Finally, since the oracle decision boundary in this example is nonlinear, the newly proposed FANS approach performs significantly better than others when the sample size is reasonably large. The above analysis seems to suggest that FANS does well  as long as we have enough data to construct accurate marginal density estimates.   Note also that ROAD is better than FANS when the training sample size is extremely small. Figure \ref{fig::toyex} shows that even under the same data distribution, the best method in practice largely depends on the available sample abundance.

A popular extension of logistic regression and close relative to FANS is the additive logistic regression, which belongs to the generalized additive model \citep{Hastie.Tibshirani.1990}. Additive logistic regression allows (smooth) nonparametric feature transformations to appear in the decision boundary by modeling
\begin{align}
  \log\frac{P(Y=1|\bX=\bx)}{P(Y=0|\bX=\bx)}=\sum_{j=1}^ph_j(x_j)\,,
\end{align}
where $h_j$'s are smooth functions. This kind of additive decision boundary is very general, in which FANS and logistic regression are special cases.  It works well for small-$p$-large-$n$ scenarios, while its penalized versions adapt to high dimensional settings.  We will compare FANS with penalized additive logistic regression in numerical studies.        The major drawback of additive logistic regression (generalized additive model) is the heavy computational complexity (e.g., the backfitting algorithm) involved in searching the transformation functions $h_j(\cdot)$. Moreover, the available algorithms, e.g., the algorithm for penGAM  \citep{Meier.Geer.ea.2009}, fail to give an estimate when  the sample size is very small.  Compared to FANS, the generalized additive model uses a factor of $K_n$ more parameters, where $K_n$ is the number of knots in the approximation of every additive components $\{h_j(\cdot)\}_{j=1}^p$.  While this reduces possible biases in comparison with FANS, it increases variances in the estimation and results in more computation cost (see Table~\ref{tb::time}).  Moreover, FANS admits a nice interpretation of optimal combination of optimal building blocks for  univariate classifiers.  


Besides the aforementioned references, there is a huge literature on high dimensional classification. 
Examples include principal component analysis in \citet{BairHast-2006} and \citet{ZouHastieTibs-2006}, partial least squares in \citet{Nguyen-2002}, \citet{HuangPan-2003} and \citet{Boulesteix-2004}, and sliced inverse regression in \citet{Li-1991-sliced-inverse-regression} and \citet{Antoniadis-2003}.  Recently, there has been a surge of interest for extending the linear discriminant analysis to high-dimensional settings including \cite{Guo.Hastie.ea.2007}, \cite{WuZhangetal2009}, \cite{Clemmensen.Hastie.ea.2011}, \cite{Shao.Wang.ea.2011}, \cite{Cai.Liu.2011}, \cite{Mai.Zou.ea.2012} and \cite{Witten.Tibshirani.2012}.

The rest of the paper is organized as follows. Section \ref{sec::algorithm} introduces the FANS algorithm.  Section \ref{sec::numerical} is dedicated to simulation studies and real data analysis. Theoretical results are presented in Section \ref{sec::excess-risk}.
We conclude with a discussion in Section \ref{sec::discussion}. Longer proofs and technical results are relegated to the Appendix.

\section{Algorithm}\label{sec::algorithm}

In this section, an efficient algorithm ($S1-S5$) for FANS will be introduced. We will also describe  a variant of FANS (FANS2), which uses the original features in addition to the transformed ones.
\subsection{FANS and its Running Time Bound}
\begin{enumerate}
\item [$S1$.] Given $n$ pairs of observations $D = \{(\bX_i,Y_i), i=1,\cdots,n\}$. Randomly split the data into two parts for $L$ times: $D_l= (D_{l1}, D_{l2})$, $l= 1, \cdots, L$. 

\item [$S2$.] On each  $D_{l1}$, $l\in\{1, \cdots, L\}$,  apply kernel density estimation and denote  the estimates by $\hat{ \bbf}=(\hat{f}_1,\cdots,\hat{f}_p)^T$ and $\hat \bbg =(\hat g_1,\cdots, \hat g_p)^T$.
\item [$S3$.] Calculate the transformed observations
$\hat{\bZ}_i=\bZ_{\hat \bbf,\hat\bbg}(\bX_i)$, where $\hat{Z}_{ij}=\log \hat{f}_j(X_{ij}) - \log \hat g_j(X_{ij})$, for each $i\in D_{l2}$ and $j\in\{1, \ldots, p\}$.
\item [$S4$.] Fit an $L_1$-penalized logistic regression to the transformed data $\{(\hat\bZ_i,Y_i), i\in D_{l2}\}$, using cross validation to get the best penalty parameter. For a new observation $
    \bx$, we estimate transformed features by $\log\hat{f}_j(x_j)-\log\hat g_j(x_j)$, $j=1, \ldots, p$, and plug them into the fitted logistic function to get the predicted probability $p_l$.
\item [$S5$.] Repeat ($S2$)-($S4$) for $l=1,\cdots,L$, use the average predicted probability $prob=L^{-1}\sum_{l=1}^L p_l$ as the final prediction,  and
    assign the observation $\bx$ to class 1 if $prob\geq 1/2$, and $0$ otherwise.

\end{enumerate}

A few comments on the technical implementation are made as follows.

\begin{remark}
\text{}
\begin{enumerate}

\item[i).] In $S2$, if an estimated marginal density is less than some threshold $\varepsilon$ (say $10^{-2}$), we set it to be $\varepsilon$. This Winsorization increases the stability of the transformations, because the estimated transformations $\log \hat{f}_j$ and $\log \hat g_j$ are unstable in regions where true densities are low.
\item[ii).]  In $S4$, we take penalized logistic regression, but any linear classifier can be used.  For example, support vector machine (SVM) with linear kernel is also a good choice.
\item[iii).] In $S4$, the $L_1$ penalty \citep{Tibshirani.1996} was adopted since our primary interest is the classification error. We can also apply other penalty functions, such as SCAD \citep{Fan.Li.2001}, adaptive LASSO \citep{Zou.2006} and MCP \citep{Zhang.2010}.
\item[iv).] In $S5$, the average predicted probability is taken as the final prediction. An alternative approach is to make a decision on each random split, and listen to majority vote.
\end{enumerate}
\end{remark}

In $S1$, we split the data multiple times.
The rationale behind multiple splitting lies in the two-step prototype nature of FANS, which  uses the first part of the data for marginal nonparametric density estimates (in $S2$) and (transformation of) the second part for penalized logistic regression (in $S4$).  Multiple splitting and prediction averaging not only make our procedure more robust against arbitrary assignments of data usage, but also make more efficient use of limited data. This idea is related to random forest \citep{breiman2001random}, where the final prediction is the average over results from multiple bootstrap samples. Other related literature includes  \cite{fu2005estimating} which considers estimation of misclassification error with small samples via bootstrap cross-validation. The number of splits is fixed at $L=20$ throughout all numerical studies. This choice reflects our cluster's node number. Interested readers can as well leverage their better computing resources for a larger $L$.  However, we observed that further increasing $L$ leads to similar performance for all simulation examples. Also, we recommend a balanced assignment by switching the role of data used for feature transformation and for feature selection, i.e.,   $D_{2l}=(D_{(2l-1), 2},D_{(2l-1), 1})$ when $D_{2l-1} = (D_{(2l-1), 1},D_{(2l-1), 2})$.


It is straightforward to derive a running time bound for our algorithm. Suppose splitting has been done.  In $S2$, we need to perform kernel density estimation for each variable, which costs $O(n^2 p)$\footnote{Approximate kernel density estimates can be computed faster, see e.g., \cite{Raykar.Duraisawami.Zhao.2010}.}. The transformations in $S3$ cost $O(np)$. In $S4$, we call the R package \pkg{glmnet} to implement penalized logistic regression, which employs  the coordinate decent algorithm for each penalty level. This step has a computational cost at most $O(npT)$, where $T$ is the number of penalty levels, i.e., the number of times the coordinate descent algorithm is run (see \cite{Friedman.Hastie.ea.2007} for a detailed analysis).    The default setting is $T=100$, though we can set it to other constants. Therefore, a running time bound for the whole algorithm is $O(L(n^2p+np+npT))=O(Lnp(n+T))$.

The above bound does not look particularly interesting. However, smart implementation of the FANS procedure can fully unleash the potential of our algorithm. Indeed, not only the $L$ repetitions, but also the marginal density estimates in $S2$ can be done via parallel computing. Suppose $L$ is the number of available nodes, and the cpu core number in each node is $N\geq n/T$. This assumption is reasonable because $T=100$ by default, $N=8$ for our implementation, and sample sizes $n$ for many applications are less than a multiple of $TN$.  Under this assumption, the $L$ predicted probabilities calculations can be carried out simultaneously and the results are combined later in $S5$.  Moreover in $S2$,  the running time bound becomes $O(n^2p/N)$.  Henceforth, a bound for the whole algorithm will be $O(npT)$, which is the same as that for penalized logistic regression. The exciting message here is that, by leveraging modern computer architecture, we are able to implement our nonparametric classification rule FANS within running time at the order of a parametric method.  The computation times for various simulation setups are reported in Table 2, where the first column reports results when only $L$ repetitions  are paralleled, and the second column reports the improvement when marginal density estimates in $S2$ are paralleled within each node.


\subsection{Augmenting Linear Features}
As we argued in the introduction, features with no marginal discrimination power do not make contribution in FANS. One remedy is to run (in $S4$) the penalized logistic regression using both the transformed features and  the original ones, which amounts to modeling the log odds by
$$
\beta_0 + \beta_1\log\frac{ f_1(x_1)}{ g_1(x_1)} + \ldots + \beta_p\log\frac{ f_p(x_p)}{ g_p(x_p)} + \beta_{p+1}x_1 + \ldots+\beta_{2p}x_p\,.
$$
This variant of FANS is named FANS2, and it allows features with no marginal power to enter the model in a linear fashion.  FANS2 helps when a linear decision boundary separates data reasonably well.

\section{Numerical Studies}\label{sec::numerical}
\subsection{Simulation}
In simulation studies,  FANS and FANS2 are compared with competing methods: penalized logistic regression (PLR, \cite{Friedman.Hastie.ea.2010}), penalized additive logistic regression models (penGAM, \cite{Meier.Geer.ea.2009}), support vector machine (SVM), regularized optimal affine discriminant (ROAD, \cite{Fan.Feng.ea.2011}), linear discriminant analysis (LDA), Naive Bayes (NB) and feature annealed independence rule (FAIR, \cite{Fan.Fan.2008}).

In all simulation settings, we set $p=1000$ and  training and testing data sample sizes of each class to be $300$. Five-fold cross-validation is conducted when needed, and we repeat 50 times for each setting (The relative small number of replications is due to the long computation time of penGAM, c.f. Table~\ref{tb::time}).
Table \ref{tb::simu-error} summarizes median test errors for each method along   with the corresponding standard errors. This table omits Fisher's classifier (using pseudo inverse for sample covariance matrix), because it gives a test error around $50\%$, equivalent to random guessing. 

\begin{example}
We consider the two class Gaussian settings where $\Sigma_{ii}=1$ for all $i=1,\cdots,p$ and $\Sigma_{ij}=\rho^{|i-j|}$,  $\bmu_1=\bzero_{1000}$ and $\bmu_2=(\bone_{10}^T,\bzero_{990}^T)^T$, in which $\bone_d$ is a length $d$ vector with all entries 1, and $\bzero_d$ is a length $d$ vector with all entries 0. Two different correlations $\rho=0$ and $\rho=0.5$ are investigated.
\end{example}

This is  the classical LDA setting. In view of the linear optimal decision boundary, the nonparametric transformations in FANS is not necessary. Table \ref{tb::simu-error} indicates some efficiency (not much) loss due to  the more complex model FANS. However, by including the original features,  FANS2 is comparable to the methods (e.g., PLR and ROAD) which learn boundaries linear in original features. In other words, the price to pay for using the more complex method FANS (FANS2) is small in terms of the classification error.

An interesting observation is that penGAM, which is based on a more general model class than FANS and FANS2, performs worse than our new methods. This is also expected as the complex parameter space considered by penGAM is unnecessary in view of a linear optimal decision boundary. Surprisingly, SVM performs poorly (even worse than NB), especially when all features are independent.

\begin{example}
The same settings as Example 1 except  the common covariance matrix is an equal correlation matrix, with a common correlation $\rho = 0.5$ and $\rho = 0.9$.
\end{example}

Same as in Example 1, FANS and FANS2 have performance comparable to PLR and ROAD.
Although FAIR works very well in Example 1, where the features are independent (or nearly independent), it fails badly when there is significant global pairwise correlation. Similar observations also hold for NB. This example shows that ignoring correlation among features could lead to significant loss of information and deterioration in the classification error.

\begin{example}
One class follows a multivariate Gaussian distribution, and the other a mixture of two multivariate Gaussian distributions. Precisely, \\
 Class 0: $ N((3\times\bone_{10}^T,\bzero_{p-10}^T)^T,\bSigma_p),$\\
 Class 1: $0.5\times N(\bzero_p,\bI_p)+0.5\times N((6\times\bone_{10}^T,\bzero_{p-10}^T)^T,\bSigma_p), $\\
where $\Sigma_{ii}=1$, $\Sigma_{ij}=\rho$ for $i\neq j$.
Correlations $\rho=0$ and $\rho=0.5$ are considered.

\end{example}

In this example,  Class 0 and Class 1 have the same mean, but have different marginal densities for the first 10 dimensions. Table \ref{tb::simu-error} shows that all methods based on linear boundary perform like random guessing, because the optimal decision boundary is highly nonlinear. penGAM is comparable to FANS and FANS2, but SVM cannot capture the oracle decision boundary well even if a nonlinear kernel is applied.

\begin{example}\label{ex::uniform-ball-square}
Two classes follow uniform distributions,\\
Class 0: $\mbox{Unif }(A)$,\\
Class 1: $\mbox{Unif }(B\backslash A)$,\\
where $A=\{\bx\in \mathbb{R}^p:\|\bx\|_2\leq 1\}$ and $B=[-1,1]^p$.
\end{example}
Clearly, the oracle decision boundary is $\{\bx\in \mathbb{R}^p:\|\bx\|_2= 1\}$. Again, FANS and FANS2 capture this simple boundary well while the linear-boundary based methods fail to do so.

Computation times (in seconds) for various classification algorithms are reported in Table \ref{tb::time}. FANS is extremely fast thanks to parallel computing. While penGAM performs similarly to FANS in the simulation examples, its computation cost is much higher. The similarity in performance is due to the abundance in training examples. We will demonstrate with an email spam classification example that penGAM fails to deliver satisfactory results on small samples.


\begin{table}[ht]
\caption{Median test error (in percentage) for the simulation examples. Standard errors are in the parentheses.\label{tb::simu-error}}
\begin{center}
\begin{tabular}{l|rrrrrrrr}
\hline
Ex($\rho$)&FANS&FANS2&ROAD&PLR&penGAM&NB&FAIR&SVM\\
\hline
 1(0)&6.8(1.1)&6.2(1.2)&6.0(1.3)&6.5(1.2)&6.6(1.1)&11.2(1.4)&5.7(1.0)&13.2(1.5)\\
 1(0.5)&16.5(1.7)&16.2(1.8)&16.5(5.3)&15.9(1.7)&16.9(1.6)&20.6(1.7)&17.2(1.6)&22.5(1.8)\\
 2(0.5)&4.2(0.9)&2.0(0.6)&2.0(0.6)&2.5(0.6)&3.7(0.9)&43.5(11.1)&25.3(1.6)&5.3(1.1)\\
 2(0.9)&3.1(1.1)&0.0(0.0)&0.0(0.0)&0.0(0.0)&0.2(1.4)&46.8(8.8)&30.2(1.9)&0.0(0.1)\\
3(0)&0.0(0.0)&0.0(0.0)&49.6(2.4)&50.0(1.3)&0.0(0.1)&50.4(2.2)&50.2(2.1)&31.8(2.4)\\
3(0.5)&3.4(0.7)&3.4(0.7)&49.3(2.4)&50.0(1.3)&3.7(0.8)&50.0(2.1)&50.2(2.0)&19.8(2.4)\\ 4&0.0(0.0)&0.0(0.0)&28.2(1.8)&50.0(10.7)&0.0(0.0)&41.0(1.1)&34.6(1.4)&0.0(0.0)\\
\hline
\end{tabular}
\end{center}
\end{table}

\begin{table}[ht]
\caption{Computation time (in seconds) comparison for FANS, SVM, ROAD and penGAM. The parallel computing technique is applied. Standard errors are in the parentheses.\label{tb::time}}
\begin{center}
\begin{tabular}{l|rrrrr}
\hline
Ex($\rho$)&FANS&FANS(para)&SVM&ROAD&penGAM\\
\hline
1(0)&12.0(2.6)&3.8(0.2)&59.4(12.8)&99.1(98.2)&243.7(151.8)\\
1(0.5)&12.7(2.1)&3.5(0.2)&81.3(19.2)&100.7(89.3)&325.8(194.3)\\
2(0.5)&16.0(3.1)&4.0(0.2)&77.6(18.1)&106.8(90.7)&978.0(685.7)\\
2(0.9)&22.0(4.6)&4.5(0.3)&75.7(17.8)&98.3(83.9)&3451.1(3040.2)\\
3(0)&12.1(2.1)&3.4(0.2)&152.1(27.4)&96.3(68.8)&254.6(130.0)\\
3(0.5)&11.9(2.0)&3.4(0.2)&342.1(58.0)&95.9(74.8)&298.7(167.4)\\
4&22.4(3.9)&6.6(0.4)&264.3(45.0)&75.1(54.0)&4811.9(3991.7)\\
\hline
\end{tabular}
\end{center}
\end{table}

\subsection{Real Data Analysis}
We study two real examples,  and compare FANS (FANS2) with competing methods.
\subsubsection{Email Spam Classification}
First, we investigate a benchmark email spam data set. This data set has been studied by \cite{Hastie.Tibshirani.ea.2009} among others to demonstrate the power of additive logistic regression models.  
There are a
total of $n=4,601$ observations with $p = 57$ numeric attributes. The attributes are, for instance, the percentage of specific words or characters in an email, the average and maximum run lengths of upper case
letters, and the total number of such letters. To show suitable application domains of FANS and FANS2, we vary the training proportion, from 5\%, 10\%, 20\%, $\cdots$,  to 80\% of the data while assigning the rest as  test set. Splits are repeated for 100 times and we report the median classification errors.

Figure \ref{fig::spam} and Table \ref{tb::spam} summarize the results. First, we notice that FANS and FANS2 are very competitive when training sample sizes are small. As the training sample size increases, SVM  becomes comparable to FANS2 and slightly better than FANS. In general, these three methods dominate throughout different training proportions.  The more complex model penGAM failed to yield classifiers when training data proportion is less than $30\%$ due to the difficulty of matrix inversion with the splines basis functions. For larger training samples, penGAM performs better than linear decision rules; however, it is not as competitive as either FANS or FANS2. Also interestingly, when the training sample size is 5\%, Naive Bayes (NB) performs as well as the sophisticated method FANS2 in terms of median classification error, but NB has a larger standard error. Moreover, the median classification error of NB remains almost unchanged when the sample size increases. In other words, NB's independence assumption allows good training given very few data points, but it cannot benefit from larger samples due to severe model bias.

\begin{figure}

  \begin{center}

\caption{The median test classification error for the spam data set using various proportions of the data as training sets for different classification methods.\label{fig::spam}}
\includegraphics[scale=0.8]{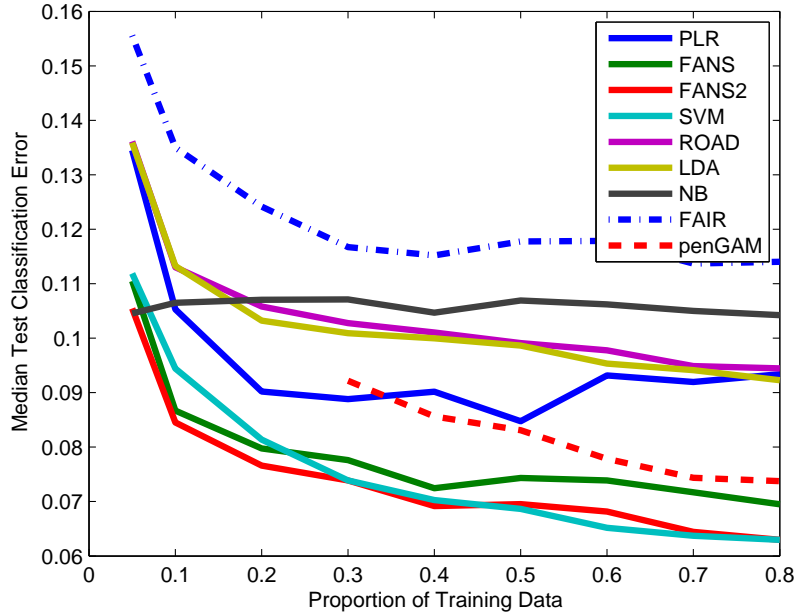}
  \end{center}
\end{figure}


\begin{table}[ht]
\caption{Median classification error (in percentage) on e-mail spam data when the size of the training data varies. Standard errors are in the parentheses.  \label{tb::spam}}

\begin{footnotesize}
\begin{center}
\begin{tabular}{l|rrrrrrrrr}
\hline
$\%$&FANS&FANS2&ROAD&PLR&penGAM&LDA&NB&FAIR&SVM\\
\hline
5&11.1(2.6)&10.5(1.1)&13.6(0.9)&13.5(1.7)&-&13.6(1.1)&10.5(5.0)&15.6(1.7)&11.2(0.8)\\
10&8.7(2.4)&8.5(0.9)&11.3(0.8)&10.5(1.1)&-&11.3(0.9)&10.7(4.2)&13.5(0.9)&9.4(0.7)\\
20&8.0(2.1)&7.7(0.7)&10.6(0.6)&9.0(0.8)&-&10.3(0.6)&10.7(5.3)&12.4(0.7)&8.1(0.7)\\
30&7.8(1.7)&7.4(0.5)&10.3(0.4)&8.9(0.6)&9.2(0.6)&10.1(0.5)&10.7(4.0)&11.7(0.4)&7.4(0.6)\\
40&7.2(2.2)&6.9(0.5)&10.1(0.5)&9.0(0.6)&8.6(0.5)&10.0(0.4)&10.5(5.1)&11.5(0.6)&7.0(0.5)\\
50&7.4(2.2)&7.0(0.5)&9.9(0.5)&8.5(0.6)&8.3(0.5)&9.9(0.4)&10.7(4.1)&11.8(0.6)&6.9(0.5)\\
60&7.4(2.2)&6.8(0.5)&9.8(0.6)&9.3(0.6)&7.8(0.6)&9.5(0.5)&10.6(4.8)&11.8(0.7)&6.5(0.6)\\
70&7.2(1.6)&6.4(0.6)&9.5(0.7)&9.2(0.7)&7.4(0.7)&9.4(0.6)&10.5(4.6)&11.4(0.7)&6.4(0.7)\\
80&6.9(1.6)&6.3(0.7)&9.4(0.6)&9.3(0.9)&7.4(0.8)&9.2(0.6)&10.4(4.7)&11.4(0.8)&6.3(0.9)\\
\hline
\end{tabular}
\end{center}

\end{footnotesize}

\end{table}

\subsubsection{Lung Cancer Classification}

We now evaluate the newly proposed classifiers on a popular gene expression data set  ``Lung Cancer" \citep{Gordon.Jensen.ea.2002}, which comes with predetermined,
separate training and test sets. It contains
$p = 12,533$ genes for $n_0 = 16$ adenocarcinoma (ADCA) and
$n_1 = 16$ mesothelioma training vectors, along with 134 ADCA
and 15 mesothelioma test vectors.

Following \cite{Dudoit.Fridlyand.ea.2002}, \citet{Fan.Fan.2008}, and \cite{Fan.Feng.ea.2011},  we standardized each sample to zero mean and unit variance. The classification results for FANS, FANS2, ROAD,  penGAM, NB, FAIR and SVM are summarized in Table \ref{tb-lung}.  FANS and FANS2 achieve 0 test classification error, while the other methods fail to do so.

\begin{table}
\caption{Classification error and number of selected genes on lung cancer data. \label{tb-lung}}
\centering
\begin{tabular}{l|llllllll}\hline
&FANS&FANS2&ROAD&PLR&penGAM&NB&FAIR&SVM\\\hline
Training Error & 0&0&1 & 0&0 &  6&  0&  0\\
  Testing Error & 0&0&1 &6& 2 & 36&  7 &4 \\
  No. of selected genes &52&52& 52 &15& 16 & 12533& 54 & 12533 \\
\hline
\end{tabular}
\end{table}

\section{Theoretical Results}\label{sec::excess-risk}
In this section, an oracle inequality  regarding the excess risk is derived for  FANS. 
Denote by $\bbf=(f_1,\cdots,f_p)^T$ and $\bbg=(g_1,\cdots, g_p)^T$  vectors of marginal densities of each class with $\bbf_0=(f_{0,1},\cdots,f_{0,p})^T$ and $\bg_0=(g_{0,1},\cdots,g_{0,p})^T$ being the true densities. Let $\{(\bX_i,Y_i)\}_{i=1}^n$ be i.i.d. copies of $(\bX,Y)$, and  the regression function be modeled by $$
P(Y_1=1|\bX_1)=\frac{1}{1+\exp(-m(\bZ_1))}\,,
$$
where $\bZ_1=(Z_{11}, \cdots, Z_{1p})^T$,  each $Z_{1j}=Z_{1j}(\bX_1)=\log f_j(X_{1j})-\log g_j(X_{1j})$, and $m(\cdot)$ is a generic function in some function class $\mathcal{M}$ that includes the linear functions.
Now, let
$\mQ=\{q=(m, \bbf, \bbg)\}$ be the parameter space of interest with constraints on $m$, $\bbf$ and $\bbg$ be specified later. The loss function we consider is
$$
\rho(q) = \rho(m, \bbf, \bbg)= \rho_q(\bX_1, Y_1)=-Y_1 m(\bZ_1)+\log(1+\exp[m(\bZ_1)])\,.
$$
Let $m_0=\arg\min_{m\in\mathcal{M}} P \rho(m,\bbf_0,\bbg_0)$. Then  the target parameter is $q^* = (m_0,\bbf_0,\bg_0)$.
We use a working model with $m_{\bbeta}(\bZ_1)=\bbeta^T\bZ_1$ to approximate $m_0$.
Under this working model, for a given parameter $q=(m_{\bbeta}, \bbf,\bbg)$,
let
\begin{align}\label{eq::logit-fans}
\pi_q(\bX_1)=P(Y_1=1|\bX_1)=\frac{1}{1+\exp(-\bbeta^T\bZ_1)}\,.
\end{align}
With this linear approximation, the loss function is the logistic  loss  $$\rho(q)= \rho_q(\bX_1, Y_1)=-Y_1\bbeta^T\bZ_1+\log(1+\exp[\bbeta^T\bZ_1])\,.$$
Denote the empirical loss by $P_n\rho(q) = \sum_{i=1}^n\rho_q(\bX_i, Y_i)/n$, and the expected loss by $P \rho(q)= E\rho_q(\bX, Y)$.
In the following, we take $\mathcal{M}$ as linear combinations of the transformed features so that  $m_0=m_{\bbeta_0}$,   where
\begin{align*}
\bbeta_0=\arg\min_{\bbeta\in\mathbb{R}^p}P\rho{(m_{\bbeta},\bbf_0,\bbg_0)}\,.
\end{align*}
In other words, $q_0=(m_{\bbeta_0},\bbf_0,\bbg_0)=q^*$.
Hence, the excess risk for a parameter $q$ is
\begin{equation} \label{eq4.4}
  \mathcal{E}(q)=P[\rho(q)-\rho(q^*)] = P[\rho(q)-\rho(q_0)]\,.
\end{equation}


As described in Section \ref{sec::algorithm}, densities $\bbf_0$ and $\bbg_0$ are unavailable and must be estimated.  Theorem \ref{thm::oracle-final} will establish the excess risk bound for the $L=1$ base procedure, which implies that the logistic regression coefficient and density estimates are close to the corresponding true values. Therefore, we expect for each $l = 1, \cdots, L$, the predicted probability $p_l$ is close to the oracle $\pi_{q^*}(\cdot)$. This further implies that $p = 1/L \sum_{l=1}^L p_l$ is close to $\pi_{q^*}(\cdot)$. Given the above analysis, we fix $L=1$ in the FANS algorithm (i.e., only one random splitting is conducted) throughout the theoretical development.

Suppose we have labeled samples $\{\bX^+_1, \cdots, \bX^+_{n_1}\}$ (used to learn $\bbf_0$) and $\{\bX^-_1, \cdots, \bX^-_{n_1}\}$ (used to learn $\bg_0$; theory carries over for different sample sizes), in addition to an i.i.d. sample $\{(\bX_1, Y_1), \cdots, (\bX_n, Y_n)\}$ (used to conduct penalized logistic regression). Moreover, suppose $\{(\bX_1, Y_1), \cdots, (\bX_n, Y_n)\}$ is  independent of $\{\bX^+_1, \cdots, \bX^+_{n_1}\}$ and $\{\bX^-_1, \cdots, \bX^-_{n_1}\}$.   A simple way to comprehend the above theoretical set up is that the sample size of $2n_1 + n$ has been split into three groups.  The notations $P$ and $E$ are regarding the random couple $(\bX, Y)$.
We use the notation $P^n$ to denote the probability measure induced by the sample $\{(\bX_1, Y_1), \cdots, (\bX_n, Y_n)\}$, and notations $P^{+}$ and  $P^{-}$ for the probability measures induced by the samples $\{\bX^+_1, \cdots, \bX^+_{n_1}\}$ and $\{\bX^-_1, \cdots, \bX^-_{n_1}\}$.

The density estimates $\hat\bbf=(\hat{f}_1, \cdots, \hat{f}_p)^T$ and $\hat\bbg=(\hat g_1, \cdots, \hat g_p)^T$ are based on samples $\{\bX^+_1, \cdots, \bX^+_{n_1}\}$ and $\{\bX^-_1, \cdots, \bX^-_{n_1}\}$:
$$
\hat{f}_j(x) = \frac{1}{n_1h}\sum_{i=1}^{n_1} K \left( \frac{X^+_{ij} - x}{h}  \right) \text{ and } \hat g_j(x) = \frac{1}{n_1h}\sum_{i=1}^{n_1} K \left( \frac{X^-_{ij} - x}{h}  \right) \text{ for } j = 1, \cdots, p\,,
$$
in which $K(\cdot)$ is a kernel function and $h$ is the bandwidth.
Then with these estimated marginal densities, we have an ``oracle estimate'' $q_1=(\bbeta_1, \hat\bbf,\hat\bbg)$,
where
\begin{eqnarray*}
\bbeta_1=\arg\min_{\bbeta\in \mbR^p} P \rho{(m_{\bbeta},\hat\bbf,\hat\bbg)}\,.
\end{eqnarray*}
It is the oracle given marginal density estimates $\hat\bbf$ and $\hat\bbg$, and is estimated in FANS by
\begin{eqnarray*}
\hat \bbeta_1=\arg\min_{\bbeta\in \mbR^p} P_n \rho{(m_{\bbeta},\hat\bbf,\hat\bbg)}+\lambda\|\bbeta\|_1\,.
\end{eqnarray*}
Let $\hat q_1 = (m_{\hat\bbeta_1}, \hat\bbf, \hat\bbg)$.
Our goal is to control the excess risk $\mathcal{E}(\hat q_1)$, where $\mathcal{E}$ is defined by \eqref{eq4.4}.
 In the following, we introduce technical conditions for this task.

Let $\bZ^0$ be the $n\times p$ design matrix consisting of transformed covariates based on the true densities $\bbf_0$ and $\bbg_0$.  That is $Z^0_{ij} = \log f_{0,j}(X_{ij})- \log g_{0,j}(X_{ij})$, for $i= 1, \cdots, n$ and $j = 1, \cdots, p$.  In addition, let $\bZ^0 = (\bZ_1^0, \bZ_2^0,\ldots, \bZ_n^0)^T$. Also, denote by $|S|$ the cardinality of the set $S$, and by $\|\bD\|_{\max} = \max_{ij} |D_{ij}|$ for any matrix $\bD$ with elements $D_{ij}$.

\begin{assumption}[Compatibility Condition]\label{assum::compatibility-condition}
The matrix $\bZ^0$ satisfies compatibility condition with a compatibility constant $\phi(\cdot)$, if for every subset $S\subset \{1,\cdots,p\}$, there exists a constant $\phi(S)$, such that for all $\bbeta\in \mbR^p$ that satisfy $\|\bbeta_{S^c}\|_1\leq 3\|\bbeta_S\|_1$, it holds that
\begin{eqnarray*}
  \|\bbeta_S\|_1^2\leq \frac{1}{n\phi^2(S)}\|\bZ^0\bbeta\|^2|S|\,.
\end{eqnarray*}
\end{assumption}
A direct application of Corollary 6.8 in \cite{Buhlmann.Geer.2011} leads to a compatibility condition on the estimated transform matrix $\hat \bZ$, in which $\hat Z_{ij} = \log \hat{f}_j(X_{ij})- \log \hat g_j(X_{ij}) $.\begin{lemma}\label{lem:compatibility-constant}
Denote by $\bE=\hat\bZ-\bZ^0$ the estimation error matrix of $\bZ^0$. If the compatibility condition is satisfied for $\bZ^0$ with a compatibility constant $\phi(\cdot)$, and the following inequalities hold
\begin{align}\label{eq::compatibility-constant}
\frac{32 \|\bE\|_{\max} |S|}{\phi(S)^2}\leq 1\,, \text{ for every } S\subset \{1, \cdots, p\},
\end{align}
the compatibility condition holds for $\hat\bZ$ with a new compatibility constant $\phi_1(\cdot)\geq \phi(\cdot)/\sqrt{2}$.
\end{lemma}

%
%

 The Compatibility Condition can be interpreted as a condition that bounds the restricted eigenvalues. The irrepresentable condition \citep{Zhao.Yu.2006} and the Sparse Riesz Condition (SRC) \citep{Zhang.Huang.2008} are in similar spirits. Essentially, these conditions avoid high correlation among subsets where signals are concentrated; such high correlation may cause difficulty in parameter estimation and risk prediction.

To help theoretical derivation, we introduce two intermediate $L_0$-penalized estimates. Given the true densities $\bbf_0$ and $\bg_0$, consider a penalized theoretical solution
$q_0^*=(\bbeta_0^*, \bbf_0,\bg_0)$, where
\begin{align}\label{eq::beta_0_star_general}
\bbeta_0^*=\arg\min_{\bbeta\in \mbR^p} 3P \rho{(m_{\bbeta},\bbf_0,\bg_0)}+2H\left(\frac{4\lambda\sqrt{s_{\bbeta}}}{\phi(S_{\bbeta})}\right)\,,
\end{align}
in which $H(\cdot)$ is a strictly convex function on $[0,\infty)$ with $H(0)=0$,  $s_{\bbeta}=|S_{\bbeta}|$ is the cardinality of $S_{\bbeta}=\{j: \beta_j\neq 0\}$, and $\phi(\cdot)$ is the compatibility constant for $\bZ^0$.
Throughout the paper, we consider a specific quadratic function\footnote{The following theoretical results can be derived for a generic strictly convex function $H(\cdot)$ along the same lines.}  $H(v)=v^2/(4c)$ whose convex conjugate is $G(u)=\sup_v\{uv-H(v)\}=cu^2$.
Then, equation \eqref{eq::beta_0_star_general} defines an $L_0$-penalized oracle:
\begin{align}\label{eq::beta_0_star}
\bbeta_0^*=\arg\min_{\bbeta\in \mbR^p} 3P \rho{(m_{\bbeta},\bbf_0,\bg_0)}+\frac{8\lambda^2 s_{\bbeta}}{c\phi^2(S_{\bbeta})}\,.
\end{align}
Similarly, with density estimate vectors $\hat\bbf$ and $\hat\bbg$, we define an $L_0$-penalized oracle estimate $q_1^*=(m_{\bbeta_1^*},\hat\bbf, \hat\bbg)$, where
\begin{align}\label{eq::beta_1_star}
\bbeta_1^*=\arg\min_{\bbeta\in \mbR^p} 3P \rho{(m_{\bbeta},\hat\bbf,\hat\bbg)}+\frac{8\lambda^2 s_{\bbeta}}{c\phi_1^2(S_{\bbeta})}\,.
\end{align}
To study the excess risk $\mathcal{E}(\hat q_1)$, we consider its relationship with $\mathcal{E}(q_1^*)$ and $\mathcal{E}(q_0^*)$.
%

\begin{assumption}[Uniform Margin Condition]\label{assum::uniform-margin-condition}
There exists $\eta>0$ such that for all $(m_{\bbeta},\bbf,\bbg)$ satisfying $\|\bbeta-\bbeta_0\|_{\infty} +\max_{1\leq j \leq p}\|f_j-f_{0,j}\|_{\infty}+\max_{1\leq j \leq p}\|g_j-g_{0,j}\|_{\infty}\leq 2 \eta$, we have
 \begin{align}
   \mathcal{E}(m_{\bbeta},\bbf,\bbg)\geq c\|\bbeta-\bbeta_0\|_2^2\,,
 \end{align}
where $c$ is the positive constant in  \eqref{eq::beta_0_star}.
\end{assumption}

The uniform margin condition is related to the one defined in \cite{Tsybakov.2004} and \cite{Geer.2008}. It is a type of ``identifiability" condition. Basically, near the target parameter $q_0=(m_{\bbeta_0},\bbf_0,\bbg_0)$, the functional value  needs to be sufficiently different from the value on $q_0$ to enable enough separability of parameters.  Note that we  impose the uniform margin condition in both the neighborhood of the parametric component $\bbeta_0$ and the nonparametric components $\bbf_0$ and $\bbg_0$, because we need to estimate the densities, in addition to the parametric part.  A related concept in binary classification is called ``Margin Assumption", which was first introduced in \cite{Polonik95} for densities.

To study the relationship between  $\mathcal{E}(\hat q_1)$ and $\mathcal{E}(q_1^*)$, we define
\begin{eqnarray*}
  v_n(\bbeta)=(P_n-P) \rho{(m_{\bbeta},\hat\bbf,\hat\bbg)}\, \mbox{ and } W_M=\sup_{\|\bbeta-\bbeta^*_1\|\leq M}|v_n(\bbeta)-v_n(\bbeta^*_1)|\,.
\end{eqnarray*}
Denote by
\begin{eqnarray*}
  2\epsilon^*=3\mathcal{E}(m_{\bbeta^*_1},\hat\bbf,\hat\bbg)+\frac{8\lambda^2s_{\bbeta^*_1}}{c\phi_1^2(S_{\bbeta^*_1})}\,.
\end{eqnarray*}
Set $M^*=\epsilon^*/\lambda_0$  ($\lambda_0$ to specified in Theorem 1)  and
\begin{eqnarray*}
  \mJ_1=\{W_{M^*}\leq \lambda_0 M^*\}=\{W_{M^*}\leq \epsilon^*\}\,.
\end{eqnarray*}
The idea here is to choose $\lambda_0$ such that the event $\mJ_1$ has high probability.

A few more notations are introduced to facilitate the discussion. Let $\tau>0$.  Denote by $\lfloor\tau\rfloor$ the largest integer strictly less than $\tau$.  For any $x, x'\in \mathbb{R}$ and any $\lfloor\tau\rfloor$ times continuously differentiable real valued function $u$ on $\mathbb{R}$, we denote by $u_x$ its Taylor polynomial of degree $\lfloor\tau\rfloor$ at point $x$:
$$
u_x(x') = \sum_{|s|\leq\lfloor \tau \rfloor}\frac{(x'-x)^s}{s!}D^s u(x)\,.
$$

For $L>0$,  the $(\tau, L, [-1,1])$-H\"{o}lder class of functions, denoted by $\Sigma(\tau, L, [-1,1])$, is the set of functions $u:\mathbb{R}\rightarrow \mathbb{R}$ that are $\lfloor\tau\rfloor$ times continuously differentiable and satisfy, for any $x,x'\in[-1,1]$, the inequality:
$$
|u(x')-u_x(x')|\leq L |x-x'|^{\tau}\,.
$$
The $(\tau, L, [-1,1])$-H\"{o}lder class of density is defined as
$$
\mathcal{P}_{\Sigma}(\tau, L, [-1,1])=\left\{p\,:\,p\geq0, \int p = 1, p\in\Sigma(\tau, L, [-1,1])\right\}\,.
$$
%

\begin{assumption}\label{assum::prob-regularity}
Assume that  $\bbeta_1$ is in the interior of  some  compact set $\mathcal{C}_p$. There exists an  $\epsilon_0\in(0,1)$ such that for all $\bbeta\in\mathcal{C}_p$ and $f_j, g_j\in\mathcal{P}_{\Sigma}(2, L, [-1, 1])$, $j = 1, \cdots, p$,
$\epsilon_0<\pi_{(m_{\bbeta},\bbf,\bbg)}(\cdot)<1-\epsilon_0$.

\end{assumption}
\begin{assumption}\label{assum::Z-max-beta-0-max}
 $\|\bZ^0\|_{\max} \leq K$ 
 for some absolute constant $K>0$, and
 $ \|\bbeta_{0}\|_\infty \leq C_1$ for some absolute constant $C_1>0$.

\end{assumption}%

\begin{assumption}\label{assum:para-constraint}

The penalty level $\lambda$ is in the range of $(8\lambda_0, L\lambda_0)$ for some $L>8$. Moreover, the following holds
 \begin{eqnarray*}
   \frac{8KL^2(e^{\eta}/\epsilon_0+1)^2}{\eta}\frac{\lambda_0s_{\bbeta^*_1}}{\phi_1^2(S_{\bbeta^*_1})}\leq 1,
 \end{eqnarray*}
 where $\eta$ is as in the uniform margin condition.
\end{assumption}

Assumption \ref{assum::prob-regularity} is a regularity condition on the probability of the event that the observation belongs to class 1.  Since the FANS estimator is based on the estimated densities, we impose the constraints in a neighborhood of the oracle estimate $\bbeta_1$ (when using $\hat\bbf$ and $\hat\bbg$).     
Assumption \ref{assum::Z-max-beta-0-max} bounds the maximum absolute entry of the design matrix as well as the maximum absolute true regression coefficient. Assumption \ref{assum:para-constraint} posits a proper range of the penalty parameter $\lambda$ to guarantee that the penalized estimator mimics the un-penalized oracle.

\begin{assumption}\label{assum::holder-class}
Suppose the feature measurement $\bX$ has a compact support $[-1,1]^p$, and $f_{0,j}, g_{0,j}\in\mathcal{P}_{\Sigma}(2, L, [-1,1])$ for all $j= 1, \cdots, p$, where $\mathcal{P}_{\Sigma}$ denotes a H\"{o}lder class of densities.

\end{assumption}

\begin{assumption}\label{assum::density-truncation}
Suppose there exists $\epsilon_l>0$ such that for all $j= 1, \cdots, p$,   $\epsilon_l \leq f_{0,j}, g_{0,j}\leq \epsilon_l^{-1}$. Also we truncate estimates $\hat{f}_j$ and $\hat g_j$ at $\epsilon_l$ and $\epsilon_l^{-1}$.
\end{assumption}

\begin{assumption}\label{assumption:vanishing}
$$
n_1^{\frac{7}{20}-\frac{3}{4}\alpha}(\log(3p))^{\frac{3}{4}}(\log n_1)^\frac{1}{10}=o(1)\,,
$$
and,
$$
n_1^{\frac{1}{10}-\alpha}(\log(3p))^{\frac{1}{2}}(\log n_1)^{\frac{2}{5}}=o(1)\,,
$$
for some constant $\alpha > 7/15$.
\end{assumption}

Assumption \ref{assum::holder-class} imposes constraints on the support of $\bX$ and smoothness condition on the true densities $\bbf_0$ and $\bbg_0$, which help control the  estimation error incurred by the nonparametric density estimates. Assumption \ref{assum::density-truncation} assumes that the marginal densities and the kernel are strictly positive on $[-1,1]^p$. Assumption \ref{assumption:vanishing} puts a restriction on the growth of the dimensionality $p$ in terms of sample size $n_1$.

We now provide a lemma to bound the uniform deviation between $\hat{f}_j$ and $f_{0,j}$ for $j=1,\cdots, p$.
\begin{lemma}\label{lemma::fhat-ghat}
Under Assumptions \ref{assum::holder-class}-\ref{assumption:vanishing}, taking the bandwidth $h = \left(\frac{\log n_1}{n_1}\right)^{1/5}$, for any $\delta_1>0$, there exists $N_1^*$ such that if $n_1\geq N_1^*$,
$$
P^{+-}\left(\max_{1\leq j\leq p} \|\hat{f}_j -  f_{0,j}\|_{\infty} \geq m\right)\leq \delta_1\,, \text{ and }
P^{+-}\left(\max_{1\leq j\leq p} \|\hat g_j -  g_{0,j}\|_{\infty} \geq m\right)\leq \delta_1\,,
$$
for $m =C_2  \sqrt{\frac{2\log(3p/\delta_1)}{n_1^{1-\alpha}}}$, and $C_2$ is an absolute constant.
\end{lemma}
 Denote by
$$
\mJ_2 = \left\{ \max_{1\leq j\leq p} \|\hat{f}_j -  f_{0,j}\|_{\infty} \leq \eta/2, \max_{1\leq j\leq p} \|\hat g_j -  g_{0,j}\|_{\infty} \leq \eta/2     \right\},
$$
where $\eta$ is the constant in the uniform margin condition.
 It is straightforward from Lemma \ref{lemma::fhat-ghat} that
$$
P^{+-}(\mJ_2)\geq 1- \frac{6p}{\exp(\eta^2 n_1^{1-\alpha}/4C_2^2)}\,.
$$
The next lemma can be similarly derived as Lemma \ref{lemma::fhat-ghat}, so its proof is omitted.
\begin{lemma}\label{uniformDev}
Under Assumptions \ref{assum::holder-class}-\ref{assumption:vanishing}, taking the bandwidth $h = \left(\frac{\log n_1}{n_1}\right)^{1/5}$, for any $\delta>0$,  there exists $N_2^{*}$ such that if $n_1\geq N_2^{*}$, \begin{eqnarray*}
P^{+-}\left( \|\bE\|_{\max} \geq m\right)\leq\delta\,,
\end{eqnarray*}
where $\bE$ is the estimation error matrix as defined in Lemma \ref{lem:compatibility-constant} and $m = C_3  \sqrt{\frac{2\log(3p/ \delta)}{n_1^{1-\alpha}}}$ for some absolute constant $C_3$.
\end{lemma}
\begin{corollary}
Under Assumptions \ref{assum::holder-class}-\ref{assumption:vanishing}, take the bandwidth $h = \left(\frac{\log n_1}{n_1}\right)^{1/5}$. On the  event $\mathcal{J}_3 = \left\{\|\bE\|_{\max} \leq C_3  \sqrt{\frac{2\log(3p/ \delta)}{n_1^{1-\alpha}}}\right\}$ (regarding labeled samples) with $P^{+-}(\mJ_3)>1-\delta$,    there exists $N_2^{*}\in\mathbb{N}$ and $C_4>0$ such that if $n_1\geq N_2^{*}$,  $|F_{kl}| = |\hat Z_{1k} - Z^0_{1k}|\cdot |\hat Z_{1l} - Z^0_{1l}|  \leq C_4 b_{n_1}$ uniformly for $k,l=1,\cdots,p$, where $b_{n_1} = 2\log (3p/ \delta)/n_1^{1-\alpha}$.
 %
\end{corollary}

\noindent Denote by
$$
\mJ_4 =\left\{32 \|\bE\|_{\max}  \max_{S\subset \{1,\ldots,p\}}\frac{|S|}{\phi(S)^2}\leq 1\right\}\,.
$$
On the event $\mJ_4$, the inequality \eqref{eq::compatibility-constant} holds, and the compatibility condition is satisfied for $\hat\bZ$ if we assume Assumption \ref{assum::compatibility-condition} (by Lemma \ref{lem:compatibility-constant}).   Moreover, it can be derived from Lemma \ref{uniformDev} by taking a specific $\delta$,
$$P^{+-}(\mJ_4)\geq 1-3p\exp\{-n_1^{1-\alpha}/(2048 C_3^2A_p^2)\}\,, $$
where $A_p= \max_{S\subset \{1,\ldots,p\}}{|S|}/{\phi(S)^2}$.
Combining  Lemma \ref{lemma::fhat-ghat} and the uniform margin condition, we see that for given estimators $\hat \bbf$ and $\hat \bbg$, the margin condition holds for the estimated transformed matrix $\hat\bZ$ involved in the FANS estimator $\hat\bbeta_1$.
Following similar lines as in  \cite{Geer.2008} delivers the following theorem, so a formal proof is omitted.

\begin{theorem}[Oracle Inequality]\label{th:1}
In addition to Assumptions \ref{assum::compatibility-condition}-\ref{assumption:vanishing},  assume $\|m_{\bbeta^*_1}-m_{\bbeta_0}\|_{\infty}\leq \eta/2$ and $\mathcal{E}(m_{\bbeta^*_1}, \hat\bbf,\hat\bbg)/\lambda_0\leq \eta/4$. Then on the event $\mJ_1\cap
\mJ_2\cap \mJ_3\cap \mJ_4$, we have
  \begin{eqnarray*}
    \mathcal{E}(m_{\hat\bbeta_1},\hat\bbf,\hat\bbg)+\lambda\|\hat\bbeta_1-\bbeta_1^*\|_1\leq 6\mathcal{E}(m_{\bbeta_1^*},\hat\bbf,\hat\bbg)+\frac{16\lambda^2{s_{\beta_1^*}}(e^{\eta}/\epsilon_0+1)^2}{c\phi_1^2(S_{\beta_1^*})}\,.
  \end{eqnarray*}
 Moreover,  when $n_1\ge \max(N_1^*,N_2^{*})$ and under the normalization condition that $\|Z_{1j}\|_{\infty}\leq 1$ for all $j = 1, \cdots, p$, it holds that
  \begin{eqnarray*}
    \mathbb{P}(\mJ_1\cap \mJ_2\cap \mJ_3\cap \mJ_4)\geq 1-\exp(-t)-6p\exp\{-\eta^2 n_1^{1-\alpha}/(4C_2^2)\}-\delta-3p\exp\{-n_1^{1-\alpha}/(2048 C_3^2A_p^2)\}\,,
  \end{eqnarray*}
 for
  \begin{eqnarray*}
    \lambda_0:=4\lambda^*+\frac{tK}{3n}+\sqrt{\frac{2t}{n}(1+8\lambda^*)}\,,
  \end{eqnarray*}
  where $\mathbb{P}$ is the probability with regards to all the samples and
  \begin{eqnarray*}
     \lambda^*=  \sqrt{\frac{2\log(2p)}{n}}+\frac{K\log (2p)}{3n}\,.
  \end{eqnarray*}

\end{theorem}
Theorem 1 shows that with high probability, the excess risk of the FANS estimator can be controlled in terms of the excess risk of $q_1^*$ when using the estimated density functions $\hat\bbf$ and $\hat\bbg$ plus a term of explicit order. Next, we will study the excess risk of $q_1^*$.
\begin{assumption}\label{assum:beta1-sparsity}
Let  $\bZ_1^0(\bbeta_1)$ be the subvector of $\bZ_1^0$ corresponding to the nonzero components of $\bbeta_1$, and $b_{n_1}=  {\log (3p/\delta_1)}/{n_1^{1-\alpha}}$.  Assume $s_{\bbeta_1}\leq a_{n_1}$ for some deterministic  sequence $\{a_{n_1}\}$, and $a_{n_1}\cdot b_{n_1}=o(1)$.   In addition, $0< C_5 \le \lambda_{\min}(P\{\bZ_1^0(\bbeta_1)\bZ_1^0(\bbeta_1)^T\})$, for some absolute constant $C_5$. 
\end{assumption}
Assumption \ref{assum:beta1-sparsity} allows the number of nonzero elements of $\bbeta_1$ to diverge at a slow rate with $n_1$. Also, it demands a lower bound of the restricted eigenvalue of the sub-matrix of $\bZ^0$ corresponding to the nonzero components of $\bbeta_1$.

\begin{lemma}\label{lejj3}
Let
$Q(\bbeta)=P\rho(m_{\bbeta},\hat\bbf,\hat\bg)+\lambda\|\bbeta\|_0$, and  $\bar{\bbeta}_1=\min\{|\bbeta_{1,j}|:\, j\in S_{\bbeta_1}\}$.
Under Assumptions \ref{assum::prob-regularity}, \ref{assum::holder-class}, \ref{assum::density-truncation}, \ref{assumption:vanishing} and \ref{assum:beta1-sparsity}, on the event $\mJ_3$, there exists a constant $N_3^*$ such that, if  $n_1\ge N_3^*$ and the penalty parameter $\lambda< 0.5 C_5\epsilon_0(1-\epsilon_0)\bar{\bbeta}_1^2$, the $L_0$ penalized solution coincides with the unpenalized version; that is
 $\bbeta_1^*=\bbeta_1$.
\end{lemma}


\begin{theorem}[Oracle Inequality]\label{thm::oracle-final}
In addition to  Assumptions \ref{assum::compatibility-condition}-\ref{assum:beta1-sparsity}, suppose $4C_1C_4s_{\bbeta_0}^2b_{n_1}\leq \lambda_0\eta$, the penalty parameter $\lambda \in (8\lambda_0, \min(L\lambda_0,0.5 C_5\epsilon_0(1-\epsilon_0)\cdot\min_{j:\beta_{1,j}\neq 0} (|\beta_{1,j}|)))$,
where $C_5$ is defined in Assumption 9, $\|m_{\bbeta^*_1}-m_{\bbeta_0}\|_{\infty}\leq \eta/2$ and $n_1\geq \max(N^*_1, N^*_2, N^*_3)$. Taking the bandwidth $h = \left(\frac{\log n_1}{ n_1}\right)^{1/5}$, on the event $\mJ_1\cap \mJ_2 \cap \mJ_3\cap \mJ_4$ as in Theorem 1, we have

\begin{equation*}
\mathcal{E}(m_{\bbeta_1^*},\hat\bbf,\hat\bbg)
\le  C_1C_4s_{\bbeta_0}^2 b_{n_1}\,.
\end{equation*}
Then in view of Theorem 1, we have
\begin{equation*}
\mathcal{E}(m_{\hat\bbeta_1},\hat\bbf,\hat\bbg)
\le  \frac{16\lambda^2{s_{\beta_1}^*}(e^{\eta}/\epsilon_0+1)^2}{c\phi_1^2(S_{\beta_1^*})}
+6C_1C_4s_{\bbeta_0}^2 b_{n_1}\,.
\end{equation*}

\end{theorem}

This theorem finale requires quite some conditions. We now de-convolute them by providing a high level description of the motivations behind these conditions.  Because FANS is essentially a two step procedure, we need both steps to do well in order to have the theoretical performance guarantee. The first step is to estimate the transformed features. In this step, we need   regularity conditions on the class conditional densities $\bbf_0$ and $\bbg_0$, and regularity conditions on the kernel density estimate components, such as the kernel $K$. Also, the sample size need to be big enough so that the kernel density estimate is close to the truth. The second step is  penalized logistic regression using the estimated transformed features. In this step, usual conditions on the penalty level, design matrix and signal strength are needed. Moreover, some conditions that link nonparametric and parametric components, i.e., the first and second steps, such as the uniform margin condition should be in place.

From Theorem 2, it is clear that the excess risk of the FANS estimator is naturally decomposed into two parts. One part is due to the nonparametric density estimation while the other part is due to the regularized logistic regression on the estimated transformed covariates. When both the penalty parameter $\lambda$ and the bandwidth $h$ of the nonparametric density estimates $\hat\bbf$ and $\hat\bbg$ are chosen appropriately, the FANS estimator will have a diminishing excess risk with high probability. Note that one can make explicit  $\lambda$ to obtain a bound on  the excess risk in terms of the sample sizes $n$ and $n_1$, and the dimensionality $p$. Also, it is worth noting that the development of oracle inequality of the FANS procedure $\hat\bbeta_1$ is accomplished via an important bridge of the $L_0$-regularized estimator $\bbeta_1^*$.

The oracle inequality   for FANS2 can be developed along similar lines. In particular, the parameter under the working model will be changed to $q_2=(m_{(\bbeta,\bgamma)},\bbf,\bbg)$ and the success probability given $\bX_1$ will be modeled by a modified logistic function
\begin{align}\label{eq::logit-fans2}
\pi_{q_2}(\bX_1)=P(Y_1=1|\bX_1)=\frac{1}{1+\exp(-\bbeta^T\bZ_1-\bgamma^T\bX_1)},
\end{align}
where we note that in addition to the transformed features, the original features are also included.  We would like to emphasize that $\bX_1$ is observed and therefore there is no need to control its estimation error as we did for $\bZ_1$. The conditions for the theory of FANS can be adapted to establish an oracle inequality for FANS2. We omit the details to avoid duplication of similar conditions and arguments.

\section{Discussion}\label{sec::discussion}

We propose a new two-step nonlinear rule FANS (and its variant FANS2) to tackle binary classification problems in high-dimensional settings. FANS first augments the original feature space by leveraging flexibility of nonparametric estimators, and then achieves feature selection through regularization (penalization). It combines linearly the best univariate transforms that essentially augment the original features for classification.
Since nonparametric techniques are only performed on each dimension, we enjoy a flexible decision boundary without suffering from the curse of dimensionality.   An array of simulation and real data examples, supported by an efficient parallelized algorithm,  demonstrate the competitive performance of  the new procedures.

To verify our methods' performance against model misspecification, we evaluate different classifiers on the following example that has non-additive optimal decision boundary. Similar to Example \ref{ex::uniform-ball-square}, FANS and FANS2 perform the best among all competing methods (penGAM performs slightly worse with a larger standard error).
\begin{example}\label{ex::nonadditive}
Non-additive decision boundary. In particular, for $\bx\sim N(\bzero_p,I_p)$, let $y = \1\{x_1^2\sqrt{x_2^2+x_3^4+1}\geq 0.75\}$.
\end{example}
\begin{table}[ht]
\caption{Median test error (in percentages) for Example \ref{ex::nonadditive}. Standard errors are in the parentheses.\label{tb::simu-error-ex5}}
\begin{center}
\begin{tabular}{l|rrrrrrrr}
\hline
Ex&FANS&FANS2&ROAD&PLR&penGAM&NB&FAIR&SVM\\
\hline
5&6.7(1.1)&6.9(1.1)&50.2(2.1)&50.0(1.4)&8.0(2.3)&50.2(2.3)&49.7(2.1)&50.0(2.0)\\
\hline
\end{tabular}
\end{center}
\end{table}

One problem in applications we are faced with is whether we should use FANS or FANS2.  While we do not have a universal rule, a rule of thumb might shed some insight. From the simulation examples, we see when the sample size is small and/or decision boundary is highly nonlinear, FANS is recommended over FANS2.  Otherwise, FANS2 is recommended. Admittedly, in real data applications, it is often impossible to know a priori how the oracle decision boundary looks like.  Data abundance can be a rough guideline in these scenarios.

A few extensions are worth further investigation. For example, an extension to multi-class classification is an interesting future work. Beyond a specific procedure, FANS establishes  a general two-step classification framework. For the first step, one can use other types of marginal density estimators, e.g., local polynomial density estimates. For the second step,  one might rely on other  classification algorithms, e.g., the support vector machine, $k$-nearest neighbors, etc. Searching for the best two-step combination is an important but difficult task, and we believe that the answer mainly depends on the specific applications.

We can further  augment the features by adding pairwise bivariate density ratios.  These bivariate densities can be approximated by the bivariate kernel density estimates.  Alternatively, we can restrict our attention to bivariate ratios of features selected by FANS.  The latter has significantly fewer features. 

Dimensions of data sets (e.g., SNPs) in many contemporary applications could be in millions. In such ultra-high dimensional scenarios, directly applying the FANS (FANS2) approach could cause problems due to high computational complexity and instability of the estimation.  It will be beneficial to have a prior step to reduce the dimensionality in the original data.  Notable works towards this effort on the theoretical front include \cite{Fan.Lv.2008}, which introduced the sure independence screening (SIS) property to screen out the marginally unimportant variables. Subsequently, \cite{Fan.Feng.ea.2011a} proposed nonparametric independence screening (NIS), an extension of SIS to the additive models.

\section{Appendix}
The appendix contains technical proofs and Lemma \ref{lejj1a}.

\begin{proof}[Proof of Lemma 2]
For any $r, m>0$,
%
\begin{eqnarray*}
&&P^{+-} \left(\max_{1\leq j\leq p}\| \hat{f}_{j} - f_{0, j}\|_{\infty}\geq m \right)\\
&\leq& e^{-rm} E^{+-} \exp\left(\max_{1\leq j\leq p}r\|\hat{f}_{j} - f_{0,j}\|_{\infty}\right)\\
&=& e^{-rm}E^{+-} \left( \max_{1\leq j\leq p}\exp r  \|\hat{f}_{j} - f_{0,j}\|_{\infty}    \right)\\
&\leq& e^{-rm} \sum_{j=1}^p E^{+-}\left( \exp r \|\hat{f}_{j} - f_{0,j}\|_{\infty}   \right)\,.
\end{eqnarray*}

Since we assumed that all $\hat{f}_j$ and  $f_{0,j}$ are uniformly bounded by $\epsilon_l^{-1}$, $\|\hat{f}_{j}-f_{0,j}\|_{\infty}$ is bounded by $\epsilon_l^{-1}$ for all $j\in\{1, \cdots, p\}$.  This coupled with  Lemma 1 in \cite{Tong.2012},  provides a high probability bound for $\|\hat{f}_{j}-f_{0,j}\|_{\infty}$, gives rise to the following inequality,

$$
E^{+-}\exp\left( r \|\hat{f}_{j} - f_{0,j}\|_{\infty}\right)
\leq \exp\left(r \sqrt{\frac{\log(n_1/\delta_2)}{n_1h}} \right) + \exp(r\epsilon_{l}^{-1})\cdot\delta_2\,,
$$
where $\delta_2$ plays the role of $\epsilon$ in Lemma 1 of \cite{Tong.2012}(taking constant $C=1$ for simplicity).

Finding the optimal order for $r$ does not seem to be feasible.  So we plug in $r = n_1^{1-\alpha}m$ and $\delta_2 = \exp(-r\epsilon_l^{-1})$, then
\begin{eqnarray*}
&&P^{+-} \left(\max_{1\leq j\leq p}\|\hat{f}_{j} - f_{0,j}\|_{\infty}\geq m \right)\\
&\leq& p\exp(-n_1^{1-\alpha}m^2)\left\{1+\exp\left( n_1^{1-\alpha}m\sqrt{\frac{\log n_1}{n_1h} + \frac{n_1^{1-\alpha}m\epsilon_l^{-1}}{n_1h}}\right)   \right\}\\
&\leq&  p\exp(-n_1^{1-\alpha}m^2)\left\{1+\exp\left[\sqrt{2}n_1^{1-\alpha}m\left(\sqrt{\frac{\log n_1}{n_1h}} +\sqrt{\frac{m\epsilon_{l}^{-1}}{n_1^{\alpha}h}}    \right)     \right]    \right\}\\
&\leq& p\exp(-n_1^{1-\alpha}m^2)\left\{ 1+\exp\left[\sqrt{2}n_1^{1-\alpha}m\left(\frac{\log n_1}{n_1}\right)^{\frac{2}{5}} +\sqrt{2}m^{\frac{3}{2}}\epsilon_l^{-\frac{1}{2}}n_1^{\frac{11}{10}-\frac{3}{2}\alpha}(\log n_1)^{\frac{1}{10}}  \right]  \right\}\,,
\end{eqnarray*}
where in the last inequality we have used the bandwidth $h = \left(\frac{\log n_1}{n_1}\right)^{1/5} $.

The results are derived by taking $m = \sqrt{\frac{2\log (3p/\delta_1)}{n_1^{1-\alpha}}}$ ( so $\delta_1 = 3p \exp(-n_1^{1-\alpha}m^2)$), and by taking
Assumption \ref{assumption:vanishing}.  Note that we need to introduce $\alpha>0$ because the consistency conditions do not hold for $\alpha = 0$. In fact, we need at least $\alpha > 7/15$.
Under this assumption, there exists a positive integer $N_1^*$ such that if $n_1\geq N_1^*$,
$$
1+ \exp\left[2^{\frac{5}{4}}\epsilon_{l}^{-\frac{1}{2}}n_1^{\frac{7}{20}-\frac{3}{4}\alpha}(\log(3p/\delta_1))^{\frac{3}{4}}(\log n_1)^\frac{1}{10} + 2 n_1^{\frac{1}{10}-\alpha}(\log(3p/\delta_1))^{\frac{1}{2}}(\log n_1)^{\frac{2}{5}}\right]    \leq 3\,.
$$
Therefore, for $n_1\geq N_1^*$,
$$
P^{+-} \left(\max_{1\leq j\leq p}\|\hat{f}_{j} - f_{0, j}\|_{\infty}\geq m \right)\leq \delta_1, \text{ for } m = \sqrt{\frac{2\log (3p/\delta_1)}{n_1^{1-\alpha}}}\,.
$$
\end{proof}



\begin{lemma}\label{lejj1a}
For any vector $\btheta_0=(\theta_{0,1},\ldots,\theta_{0,p})^T$, let
$S_{\btheta_0}=\{j:\, \theta_{0,j}\neq 0\}$,
and let the minimum signal level be $\bar{\btheta}_0=\min\{|\theta_{0,j}|:\, j\in S_{\btheta_0}\}$.
Let
$g(\theta_j)=c_j(\theta_j-\theta_{0,j})^2+\lambda\|\theta_{j}\|_0,$
where $c_j>0$.
 If
$\lambda\le c_j\bar{\btheta}_0^2$,
$g(\theta_j)$ achieves the unique minimum at $\theta_j=\theta_{0,j}$.
\end{lemma}

\begin{proof}[Proof of Lemma \ref{lejj1a}]

For $\theta_{0,j}=0$, the result is obvious.
For $\theta_{0,j}\neq 0$, we have $j\in S_{\btheta_0}$ and
\begin{eqnarray*}
g(\theta_j)
&\ge& \lambda \|\theta_j\|_0 I(\theta_j\ne 0) +c_j (\theta_j-\theta_{0,j})^2 I(\theta_j=0)\\
&=& \lambda \|\theta_j\|_0 I(\theta_j\ne 0) +c_j \theta_{0,j}^2 I(\theta_j=0).
\end{eqnarray*}
If $\lambda\|\theta_{0,j}\|_0\le c_j\bar{\btheta}_0^2$,
\begin{eqnarray*}
g(\theta_j)
\ge \lambda \|\theta_j\|_0 I(\theta_j\ne 0)+\lambda\|\theta_{0,j}\|_0 I(\theta_j=0)
=\lambda\|\theta_{0,j}\|_0.
\end{eqnarray*}
Since $g(\theta_{0,j})=\lambda\|\theta_{0,j}\|_0$,
the lemma follows.

\end{proof}

\begin{proof}[Proof of Lemma \ref{lejj3}]


Denote
$Q_0(\bbeta)=P\rho(m_{\bbeta},\hat\bbf,\hat\bg)$. Then we have
$\bbeta_1=\arg\min_{\bbeta\in\mbR^p}Q_0(\bbeta).$ Since  $\nabla Q_0(\bbeta_1)=0$ and
$$\nabla^2Q_0(\bbeta)=P\{\hat{\bZ}_1\hat{\bZ}_1^T\exp(\hat{\bZ}_1^T\bbeta)(1+\exp(\hat{\bZ}_1^T\bbeta))^{-2}\}
\geq \epsilon_0(1-\epsilon_0) P\{\hat{\bZ}_1\hat{\bZ}_1^T\}\succeq {\mathbf 0}.$$
By Taylor's expansion of $Q_0(\bbeta)$ at $\bbeta_1$,
\begin{eqnarray}\label{eqjj8a}
Q(\bbeta)=Q_0(\bbeta_1)
+0.5(\bbeta-\bbeta_1)^T\nabla^2 Q_0(\tilde{\bbeta})(\bbeta-\bbeta_1)+\lambda \|\bbeta\|_0\, ,
\end{eqnarray}
where
$\tilde{\bbeta}$ lies between $\bbeta$ and $\bbeta_1$.
Let $\widehat{\bM}=P\{\hat{\bZ}_1(\bbeta_1)\hat{\bZ}_1(\bbeta_1)^T\}$, where $\hat \bZ_1(\bbeta_1)$ is the subvector of $\hat \bZ_1$ corresponding to the nonzero components of $\bbeta_1$,
and
$\bM=P\{\bZ_1^0(\bbeta_1)\bZ_1^0(\bbeta_1)^T\}$,
where $\bZ_1^0(\bbeta_1)$ is the subvector of $\bZ_1^0$ corresponding to the nonzero components of $\bbeta_1$.
Let $\bold{F}=\widehat{\bM}-\bM$ (a symmetric matrix). From the uniform deviance result of Lemma 3,  with probability $1-\delta$ regarding the labeled samples,  there exists a constant $C_4>0$ such that $|F_{kl}|\leq C_4 b_{n_1}$ uniformly for $k,l=1,\cdots,s_{\bbeta_1}$, where $b_{n_1} = {2\log (3p/ \delta)}/{n_1^{1-\alpha}}$.

Hence, $\|\bold{F}\|_2\leq \|\bold{F}\|_F\leq C_4 s_{\bbeta_1}b_{n_1}\leq C_4 a_{n_1}b_{n_1}$.
For any eigenvalue $\lambda(\widehat{\bM})$,
by the Bauer-Fike inequality (Bhatia, 1997), we have $\min_{1\le k\le s_{\bbeta_1}}|\lambda(\widehat{\bM})-\lambda_k(\bM)|\le \|\bold{F}\|_2\leq C_4 a_{n_1}b_{n_1},$
where $\lambda_k(\bA)$ denotes the $k$-th largest eigenvalue of $\bA$.
In addition, in view of  Assumption \ref{assum:beta1-sparsity}, there exists  $k\in S_{\bbeta_1}$ such that
$$\lambda_{\min}(\widehat{\bM})\geq \lambda_k(\bM)- C_4 a_{n_1}b_{n_1}\ge \lambda_{\min}(\bM)- C_4 a_{n_1}b_{n_1} \ge C_5 - C_4 a_{n_1}b_{n_1}.$$
Since $a_{n_1}b_{n_1}=o(1)$, there exists $N_3^*(\delta)$  such that when $n_1>N_3^*(\delta)$, we have $\lambda_{\min}(\widehat{\bM})>0$.

 Let $\bbeta_1^{(1)}$ be the subvector of $\bbeta_1$ consisting of the nonzero components.
Then by (\ref{eqjj8a}) and  Lemma~\ref{lejj1a} for each $j\in S_{\bbeta_1}$ with $\lambda < 0.5C_5\epsilon_0(1-\epsilon_0)\bar{\bbeta_1}^2$, we have
\begin{eqnarray}\label{eqjj9a}
Q(\bbeta)
&\ge& Q_0(\bbeta_1)
+0.5(C_5-C_4a_{n_1}b_{n_1})\epsilon_0(1-\epsilon_0)\|\bbeta^{(1)}-\bbeta_1^{(1)}\|^2+\lambda \|\bbeta\|_0\nonumber\\
&\ge& Q_0(\bbeta_1)+\sum_{j\in S_{\bbeta_1}}\bigl\{0.5(C_5-C_4a_{n_1}b_{n_1})\epsilon_0(1-\epsilon_0)(\beta_j-\beta_{1,j})^2+\lambda\|\beta_{j}\|_0\bigr\},
\end{eqnarray}
where
$\beta_j$ and $\beta_{1,j}$ are the $j$-th components of $\bbeta$ and $\bbeta_1$, respectively.
 For $n_1\geq N^*_{3}(\delta)$,
 \begin{eqnarray*}
Q(\bbeta)\ge Q_0(\bbeta_1)+\lambda\sum_{j\in S_{\bbeta_1}}\|\beta_{1,j}\|_0
=Q_0(\bbeta_1)+\lambda\|\bbeta_1\|_0\,.
\end{eqnarray*}By (\ref{eqjj8a}), we have
$$Q(\bbeta_1)=Q_0(\bbeta_1)+\lambda\|\bbeta_1\|_0\,.$$
Therefore, $\bbeta_1$ is a local minimizer of $Q(\bbeta)$. It then follows from
the convexity of $Q(\bbeta)$ that $\bbeta_1$ is the global minimizer $\bbeta_1^*$ of $Q(\bbeta)$.
\end{proof}

\begin{proof}[Proof of Theorem \ref{thm::oracle-final}]
For simplicity, denote by $\rho(m(\bZ_1),Y_1)$  the loss function
$\rho_q(\bX_1,Y_1)=-Y_1 m(\bZ_1)+\log(1+\exp(m(\bZ_1))$.
Note that
$$\frac{\partial \rho(m(\bZ_1),Y_1)}{\partial m(\bZ_1)}=-Y_1+\frac{\exp(m(\bZ_1))}{1+\exp(m(\bZ_1))}=-Y_1+\pi_{m, \bbf_0,\bbg_0}(\bX_1),$$
and
\begin{eqnarray*}
\frac{\partial^2 \rho(m(\bZ_1),Y_1)}{[\partial m(\bZ_1)]^2}
=\frac{\exp(m(\bZ_1))}{[1+\exp(m(\bZ_1))]^2}.
\end{eqnarray*}
By the second order Taylor expansion, we obtain that
\begin{eqnarray}
\rho(m_{\bbeta}(\hat{\bZ}_1),Y_1)
&=&\rho(m_{\bbeta_0}(\bZ^0_1),Y_1)+[\partial\rho(m_{\bbeta_0}(\bZ^0_1),Y_1)/\partial m_{\bbeta}(\bZ_1)] (m_{\bbeta}(\hat{\bZ}_1)-m_{\bbeta_0}(\bZ^0_1))\nonumber\\
&&+\frac{1}{2}\frac{\partial^2\rho(m^*,Y_1)}{[\partial m_{\bbeta}(\bZ_1)]^2} (m_{\bbeta}(\hat{\bZ}_1)-m_{\bbeta_0}(\bZ_1))^2,  \label{eqjj1}
\end{eqnarray}
where $m^*$ lies between $m_{\bbeta}(\hat{\bZ}_1)$ and $m_{\bbeta_0}(\bZ^0_1)$.
Since
\begin{equation}\label{eqjj2}
P\left[\frac{\partial\rho(m_{\bbeta_0}(\bZ^0_1),Y_1)}{\partial m_{\bbeta}(\bZ_1)}\right]=0
\end{equation}
and $0<\partial^2 \rho(m^*,Y_1)/[\partial m_{\bbeta}(\bZ_1)]^2<1,$
taking the expectation we obtain that
\begin{eqnarray*}
|P\rho(m_{\bbeta}(\hat{\bZ}_1),Y_1)
-P\rho(m_{\bbeta_0}(\bZ^0_1),Y_1)|
&<& 0.5 P[(m_{\bbeta}(\hat{\bZ}_1)-m_{\bbeta_0}(\bZ^0_1))^2]\\
&=& 0.5P[(\hat{\bZ}_1^T\bbeta-(\bZ^0_1)^T\bbeta_0)^2].
\end{eqnarray*}
Hence, from Corollary 1, on the event $\mJ_3$,
\begin{eqnarray*}
|P\rho(m_{\bbeta_0}(\hat{\bZ}_1),Y_1)
-P\rho(m_{\bbeta_0}(\bZ^0_1),Y_1)|
&\leq& 0.5\bbeta_0^TP[(\hat{\bZ}_1-\bZ^0_1)(\hat{\bZ}_1-\bZ^0_1)^T]\bbeta_0\\
&\leq& C_1 C_4s_{\bbeta_0}^2b_{n_1},
\end{eqnarray*}
where $s_{\bbeta}=|S_{\bbeta}|$ is the cardinality of
$S_{\bbeta}=\{j:\, \beta_j\neq 0\}$.
Naturally,
$P\rho(m_{\bbeta_0}(\hat{\bZ}_1),Y_1)
\leq P\rho(m_{\bbeta_0}(\bZ^0_1),Y_1)+C_1C_4s_{\bbeta_0}^2b_{n_1}$.

In addition,  by definition of $\bbeta_1$, $P\rho(m_{\bbeta_1}(\hat{\bZ}_1),Y_1)=\min_{\bbeta} P\rho(m_{\bbeta}(\hat{\bZ}_1),Y_1)$. As a result, $P\rho(m_{\bbeta_1}(\hat{\bZ}_1),Y_1)\le P\rho(m_{\bbeta_0}(\hat{\bZ}_1),Y_1).$
Thus, we have
\begin{equation}\label{eqjj3}
P\rho(m_{\bbeta_1}(\hat{\bZ}_1),Y_1)\le
P\rho(m_{\bbeta_0}(\bZ^0_1),Y_1)+C_1C_4s_{\bbeta_0}^2b_{n_1}.
\end{equation}

In addition, by (\ref{eqjj1}) and (\ref{eqjj2}),
for any $\bbeta$ we have
$P\rho(m_{\bbeta}(\hat{\bZ}_1),Y_1)\ge P\rho(m_{\bbeta_0}(\bZ^0_1),Y_1)$. Then, setting  $\bbeta=\bbeta_1$ on the left side leads to
\begin{equation}\label{eqjj4}
P\rho(m_{\bbeta_1}(\hat{\bZ}_1),Y_1)\ge P\rho(m_{\bbeta_0}(\bZ^0_1),Y_1).
\end{equation}
Combining (\ref{eqjj3}) and (\ref{eqjj4}) leads to
\begin{equation}\label{eqjj5}
|P\rho(m_{\bbeta_1}(\hat{\bZ}_1),Y_1)-
P\rho(m_{\bbeta_0}(\bZ^0_1),Y_1)|\leq C_1C_4s_{\bbeta_0}^2b_{n_1}.
\end{equation}
As a result, we have
\begin{align}\label{eq::excess-risk-beta1-fhat}
\mathcal{E}(m_{\bbeta_1},\hat\bbf,\hat\bg)\leq
C_1C_4s_{\bbeta_0}^2b_{n_1}.
\end{align}
\eqref{eq::excess-risk-beta1-fhat} combined with Lemma~\ref{lejj3} ($\bbeta_1^*=\bbeta_1$) leads to
\begin{equation}\label{eqjj6}
\mathcal{E}(m_{\bbeta_1^*},\hat\bbf,\hat\bg)\leq C_1C_4s_{\bbeta_0}^2b_{n_1}.
\end{equation}
Recall the oracle estimator
$$\bbeta_1^*=\arg\min_{\bbeta\in \mathcal{B}}
\Bigl\{\mathcal{E}(m_{\bbeta},\hat\bbf,\hat\bg)+\frac{8\lambda s_{\beta}}{c\phi_1^2(S_{\beta})}\Bigr\}.$$
Then by Theorem~1,
\begin{equation} \label{eqjj7}
\mathcal{E}(\hat{q}_1)
=\mathcal{E}(m_{\hat{\bbeta}_1}, \hat\bbf,\hat\bg)
\le 6 \mathcal{E}(m_{\bbeta_1^*}, \hat\bbf,\hat\bg)
+\frac{16\lambda^2{s_{\beta_1^*}}(e^{\eta}/\epsilon_0+1)^2}{c\phi^2(S_{\beta_1^*})}.
\end{equation}
Therefore, by (\ref{eqjj6}) and (\ref{eqjj7}),
\begin{eqnarray*}
\mathcal{E}(\hat{q}_1)
&\le& \frac{16\lambda^2{s_{\beta_1^*}}(e^{\eta}/\epsilon_0+1)^2}{c\phi^2(S_{\beta_1^*})}+C_1C_4s_{\bbeta_0}^2b_{n_1}.
\end{eqnarray*}

\end{proof}

\bibliographystyle{ims}
\bibliography{FANS-ref}

\end{document}